\documentclass[10pt,conference]{IEEEtran}
\usepackage{fixltx2e}
\usepackage{cite}
\usepackage{url}
\usepackage{mathrsfs}
\usepackage{color}
\usepackage{float}
\usepackage[caption=false]{subfig}

\ifCLASSINFOpdf
   \usepackage[pdftex]{graphicx}
   \graphicspath{{Figs/}}
   \DeclareGraphicsExtensions{.pdf,.jpeg,.png}
\else
\fi

\usepackage[cmex10]{amsmath}
\usepackage{amsmath}
\usepackage{amssymb}
\usepackage{amsthm}
\usepackage{amsfonts}
\usepackage{bm}
\usepackage{xfrac}
\usepackage{empheq}
\usepackage[normalem]{ulem} 
\usepackage{soul} 
\usepackage{mathtools}

\newtheorem{theorem}{Theorem}

\newtheorem{example}{Example}
\newtheorem{proposition}{Proposition}
\newtheorem{lemma}{Lemma}

\newtheorem{corollary}{Corollary}
\newtheorem{remark}{Remark}
\theoremstyle{definition}
\newtheorem{definition}{Definition}

\usepackage[a4paper,bindingoffset=0.2in,%
left=1in,right=1in,top=1in,bottom=1in,%
footskip=.25in]{geometry}
\graphicspath{{figs/}}

\begin{document}
	\newgeometry{left=0.7in,right=0.7in,top=.5in,bottom=1in}
	\title{Bounds for Privacy-Utility Trade-off with Non-zero Leakage}
\vspace{-5mm}
\author{
		\IEEEauthorblockN{Amirreza Zamani, Tobias J. Oechtering, Mikael Skoglund \vspace*{0.5em}
			\IEEEauthorblockA{\\
                              Division of Information Science and Engineering, KTH Royal Institute of Technology \\
				Email: \protect amizam@kth.se, oech@kth.se, skoglund@kth.se }}
		}
	\maketitle

\begin{abstract}
	The design of privacy mechanisms for two scenarios is studied where the private data is hidden or observable. 
	In the first scenario, 
	an agent observes useful data $Y$, which is correlated with private data $X$, and wants to disclose the useful information to a user. A privacy mechanism is employed to generate data $U$ that maximizes the revealed information about $Y$ while satisfying a privacy criterion. 
	In the second scenario, the agent has additionally access to the private data. 
	To this end, the Functional Representation Lemma and Strong Functional Representation Lemma are extended relaxing the independence condition and thereby allowing a certain leakage. 
	Lower bounds on privacy-utility trade-off are derived for the second scenario as well as upper bounds for both scenarios. 
	In particular, for the case where no leakage is allowed, our upper and lower bounds improve previous bounds.
\end{abstract}
\section{Introduction}

In this paper, random variable (RV) $Y$ denotes the useful data and is correlated with the private data denoted by RV $X$. Furthermore, disclosed data is described by RV $U$. Two scenarios are considered in this work, where in both scenarios, an agent wants to disclose the useful information to a user as shown in Fig.~\ref{ISITsys}. In the first scenario, the agent observes $Y$ and has not directly access to $X$, i.e., the private data is hidden. The goal is to design $U$ based on $Y$ that reveals as much information as possible about $Y$ and satisfies a privacy criterion. We use mutual information to measure utility and privacy leakage. In this work, some bounded privacy leakage is allowed, i.e., $I(X;U)\leq \epsilon$. In the second scenario, the agent has access to both $X$ and $Y$ and can design $U$ based on $(X,Y)$ to release as much information as possible about $Y$ while satisfying the bounded leakage constraint.
\\
The privacy mechanism design problem is receiving increased attention in information theory recently. Related works can be found in 
\cite{Calmon2,yamamoto, sankar,borz, gun,khodam,Khodam22,kostala,issa, makhdoumi, dwork1, calmon4, issajoon, asoo, Total, issa2}. 
In \cite{Calmon2}, fundamental limits of the privacy utility trade-off measuring the leakage using estimation-theoretic guarantees are studied.
In \cite{yamamoto}, a source coding problem with secrecy is studied.
\begin{figure}[]
	\centering
	\includegraphics[scale = .15]{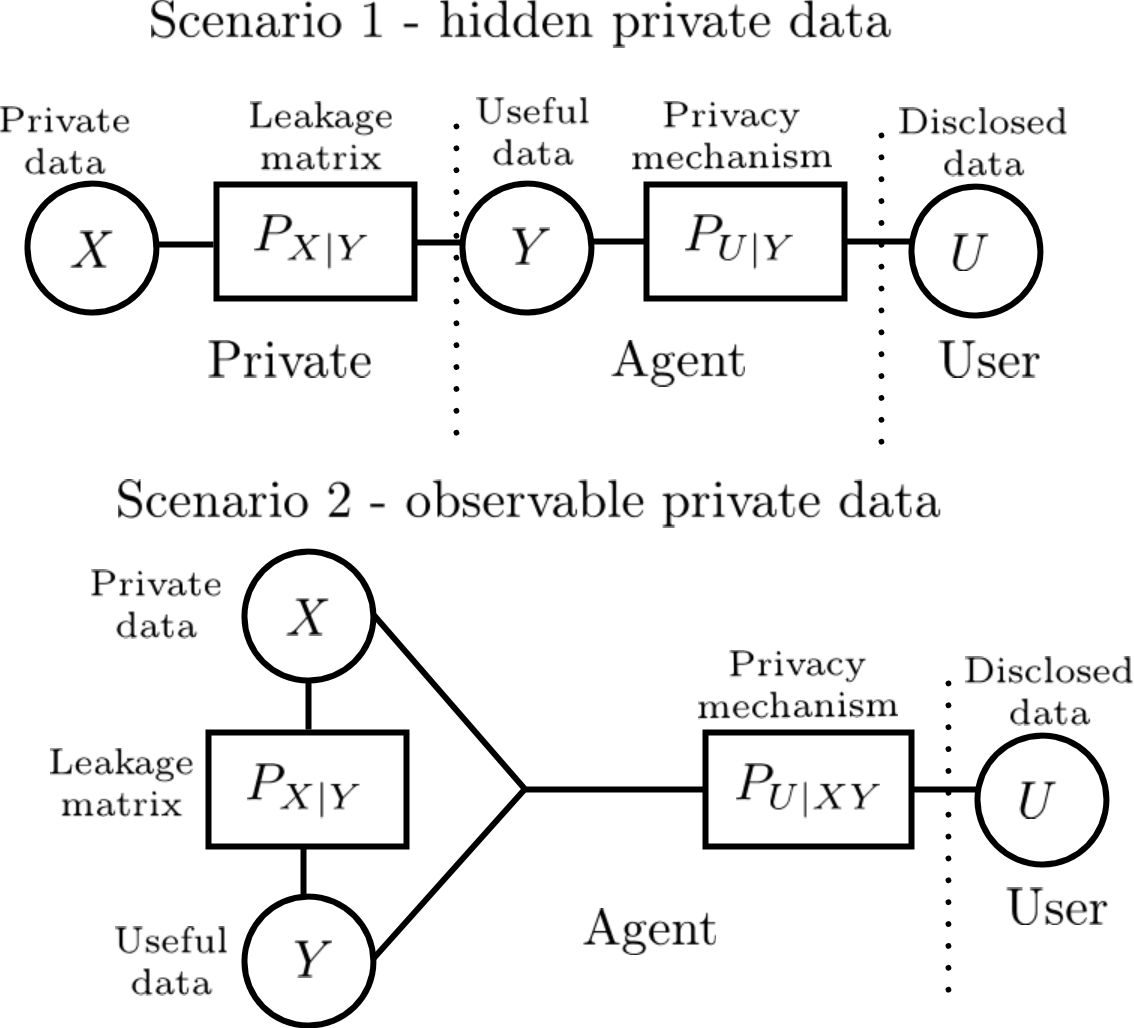}
	\caption{Considered two scenarios where the agent has only access to $Y$ and the agent has access to both $X$ and $Y$.}
	\label{ISITsys}
\end{figure}
Privacy-utility trade-offs considering equivocation as measure of privacy and expected distortion as a measure of utility are studied in both \cite{yamamoto} and \cite{sankar}.
In \cite{borz}, the problem of privacy-utility trade-off considering mutual information both as measures of privacy and utility given the Markov chain $X-Y-U$ is studied. It is shown that under perfect privacy assumption, i.e., $\epsilon=0$, the privacy mechanism design problem can be reduced to a linear program. This work has been extended in \cite{gun} considering the privacy utility trade-off with a rate constraint for the disclosed data.
Moreover, in \cite{borz}, it has been shown that information can be only revealed if $P_{X|Y}$ is not invertible. In \cite{khodam}, we designed privacy mechanisms with a per letter privacy criterion considering an invertible $P_{X|Y}$ where a small leakage is allowed. We generalized this result to a non-invertible leakage matrix in \cite{Khodam22}.

Our problem here is closely related to \cite{kostala}, where the problem of \emph{secrecy by design} is studied. Similarly, the two scenarios are considered while the results are however derived under the perfect secrecy assumption, i.e., no leakages are allowed which corresponds to $\epsilon=0$. Bounds on secure decomposition have been derived using the Functional Representation Lemma and new bounds on privacy-utility trade-off for the two scenarios are derived. The bounds are tight when the private data is a deterministic function of the useful data. 

In the present work, we generalize the privacy problems considered in \cite{kostala} by relaxing the perfect privacy constraint and allowing some leakages. To this end, we extend the Functional Representation Lemma relaxing the independence condition. Additionally, we derive bounds by also extending the Strong Functional Representation Lemma, introduced in \cite{kosnane}. Furthermore, in the special case of perfect privacy we find a new upper bound for perfect privacy function by using the \emph{excess functional information} introduced in \cite{kosnane}. We then compare our new lower and upper bounds with the bounds found in \cite{kostala} when the leakage is zero.

\section{system model and Problem Formulation} \label{sec:system}
Let $P_{XY}$ denote the joint distribution of discrete random variables $X$ and $Y$ defined on alphabets $\cal{X}$ and $\cal{Y}$. We assume that cardinality $|\mathcal{X}|$ is finite and $|\mathcal{Y}|$ is finite or countably infinite.
We represent $P_{XY}$ by a matrix defined on $\mathbb{R}^{|\mathcal{X}|\times|\mathcal{Y}|}$ and 
marginal distributions of $X$ and $Y$ by vectors $P_X$ and $P_Y$ defined on $\mathbb{R}^{|\mathcal{X}|}$ and $\mathbb{R}^{|\mathcal{Y}|}$ given by the row and column sums of $P_{XY}$. 
We represent the leakage matrix $P_{X|Y}$ by a matrix defined on $\mathbb{R}^{|\mathcal{X}|\times|\cal{Y}|}$.

For both design problems we use mutual information as utility and leakage measures. The privacy mechanism design problems for the two scenarios can be stated as follows 
\begin{align}
g_{\epsilon}(P_{XY})&=\sup_{\begin{array}{c} 
	\substack{P_{U|Y}:X-Y-U\\ \ I(U;X)\leq\epsilon,}
	\end{array}}I(Y;U),\label{main2}\\
h_{\epsilon}(P_{XY})&=\sup_{\begin{array}{c} 
	\substack{P_{U|Y,X}: I(U;X)\leq\epsilon,}
	\end{array}}I(Y;U).\label{main1}
\end{align} 
The relation between $U$ and $Y$ is described by the kernel $P_{U|Y}$ defined on $\mathbb{R}^{|\mathcal{U}|\times|\mathcal{Y}|}$, furthermore, the relation between $U$ and the pair $(Y,X)$ is described by the kernel $P_{U|Y,X}$ defined on $\mathbb{R}^{|\mathcal{U}|\times|\mathcal{Y}|\times|\mathcal{X}|}$.
 The function $h_{\epsilon}(P_{XY})$ is used when the privacy mechanism has access to both the private data and the useful data. The function $g_{\epsilon}(P_{XY})$ is used when the privacy mechanism has only access to the useful data. Clearly, the relation between $h_{\epsilon}(P_{XY})$ and $g_{\epsilon}(P_{XY})$ can be stated as follows
 \begin{align}
 g_{\epsilon}(P_{XY})\leq h_{\epsilon}(P_{XY}).
 \end{align}
 In the following we study the case where $0\leq\epsilon< I(X;Y)$, otherwise the optimal solution of $h_{\epsilon}(P_{XY})$ or $g_{\epsilon}(P_{XY})$ is $H(Y)$ achieved by $U=Y$. 
 \begin{remark}
 	\normalfont
 	For $\epsilon=0$, \eqref{main2} leads to the perfect privacy problem studied in \cite{borz}. It has been shown that for a non-invertible leakage matrix $P_{X|Y}$, $g_0(P_{XY})$ can be obtained by a linear program.
 \end{remark}
 \begin{remark}
 	\normalfont
 	For $\epsilon=0$, \eqref{main1} leads to the secret-dependent perfect privacy function $h_0(P_{XY})$, studied in \cite{kostala}, where upper and lower bounds on $h_0(P_{XY})$ have been derived. 
 \end{remark}

 \section{Main Results}\label{sec:resul}
 In this section, we first recall the Functional Representation Lemma (FRL) \cite[Lemma~1]{kostala} and Strong Functional Representation Lemma (SFRL) \cite[Theorem~1]{kosnane} for discrete $X$ and $Y$. Then we extend them for correlated $X$ and $U$, i.e., $0\leq I(U;X)=\epsilon$ and we call them Extended Functional Representation Lemma (EFRL) and Extended Strong Functional Representation Lemma (ESFRL), respectively.
 \begin{lemma}\label{lemma1} (Functional Representation Lemma \cite[Lemma~1]{kostala}):
 	For any pair of RVs $(X,Y)$ distributed according to $P_{XY}$ supported on alphabets $\mathcal{X}$ and $\mathcal{Y}$ where $|\mathcal{X}|$ is finite and $|\mathcal{Y}|$ is finite or countably infinite, there exists a RV $U$ supported on $\mathcal{U}$ such that $X$ and $U$ are independent, i.e., we have
 	\begin{align}\label{c1}
 	I(U;X)=0,
 	\end{align}
 	$Y$ is a deterministic function of $(U,X)$, i.e., we have
 	\begin{align}
 	H(Y|U,X)=0,\label{c2}
 	\end{align}
 	and 
 	\begin{align}
 	|\mathcal{U}|\leq |\mathcal{X}|(|\mathcal{Y}|-1)+1.\label{c3}
 	\end{align}
 \end{lemma}
 \begin{lemma}\label{lemma2} (Strong Functional Representation Lemma \cite[Theorem~1]{kosnane}):
 	For any pair of RVs $(X,Y)$ distributed according to $P_{XY}$ supported on alphabets $\mathcal{X}$ and $\mathcal{Y}$ where $|\mathcal{X}|$ is finite and $|\mathcal{Y}|$ is finite or countably infinite with $I(X,Y)< \infty$, there exists a RV $U$ supported on $\mathcal{U}$ such that $X$ and $U$ are independent, i.e., we have
 	\begin{align*}
 	I(U;X)=0,
 	\end{align*}
 	$Y$ is a deterministic function of $(U,X)$, i.e., we have 
 	\begin{align*}
 	H(Y|U,X)=0,
 	\end{align*}
 	$I(X;U|Y)$ can be upper bounded as follows
 	\begin{align*}
 	I(X;U|Y)\leq \log(I(X;Y)+1)+4,
 	\end{align*}
 	and 
 	$
 	|\mathcal{U}|\leq |\mathcal{X}|(|\mathcal{Y}|-1)+2.
 	$
 \end{lemma}
 \begin{remark}
 	By checking the proof in \cite[Th.~1]{kosnane}, the term $e^{-1}\log(e)+2+\log(I(X;Y)+e^{-1}\log(e)+2)$ can be used instead of $\log(I(X;Y)+1)+4$. 
 \end{remark}
 \begin{lemma}\label{lemma3} (Extended Functional Representation Lemma):
 	For any $0\leq\epsilon< I(X;Y)$ and pair of RVs $(X,Y)$ distributed according to $P_{XY}$ supported on alphabets $\mathcal{X}$ and $\mathcal{Y}$ where $|\mathcal{X}|$ is finite and $|\mathcal{Y}|$ is finite or countably infinite, there exists a RV $U$ supported on $\mathcal{U}$ such that the leakage between $X$ and $U$ is equal to $\epsilon$, i.e., we have
 	\begin{align*}
 	I(U;X)= \epsilon,
 	\end{align*}
 	$Y$ is a deterministic function of $(U,X)$, i.e., we have  
 	\begin{align*}
 	H(Y|U,X)=0,
 	\end{align*}
 	and 
 	$
 	|\mathcal{U}|\leq \left[|\mathcal{X}|(|\mathcal{Y}|-1)+1\right]\left[|\mathcal{X}|+1\right].
 	$
 \end{lemma}
 \begin{proof}
 	Let $\tilde{U}$ be the RV found by FRL and let $W=\begin{cases}
 	X,\ \text{w.p}.\ \alpha\\
 	c,\ \ \text{w.p.}\ 1-\alpha
 	\end{cases}$, where $c$ is a constant which does not belong to the support of $X$ and $Y$ and $\alpha=\frac{\epsilon}{H(X)}$. We show that $U=(\tilde{U},W)$ satisfies the conditions. We have
 	\begin{align*}
 	I(X;U)&=I(X;\tilde{U},W)\\&=I(\tilde{U};X)+I(X;W|\tilde{U})\\&\stackrel{(a)}{=}H(X)-H(X|\tilde{U},W)\\&=H(X)-\alpha H(X|\tilde{U},X)-(1-\alpha)H(X|\tilde{U},c)\\&=H(X)-(1-\alpha)H(X)=\alpha H(X)=\epsilon,	
 	\end{align*}
 	where in (a) we used the fact that $X$ and $\tilde{U}$ are independent. Furthermore,
 	\begin{align*}
 	H(Y|X,U)&=H(Y|X,\tilde{U},W)\\&=\alpha H(Y|X,\tilde{U})+(1-\alpha)H(Y|X,\tilde{U},c)\\&=H(Y|X,\tilde{U})=0.
 	\end{align*}
 	In the last line we used the fact that $\tilde{U}$ is produced by FRL.
 \end{proof}
 \begin{lemma}\label{lemma4} (Extended Strong Functional Representation Lemma):
 	For any $0\leq\epsilon< I(X;Y)$ and pair of RVs $(X,Y)$ distributed according to $P_{XY}$ supported on alphabets $\mathcal{X}$ and $\mathcal{Y}$ where $|\mathcal{X}|$ is finite and $|\mathcal{Y}|$ is finite or countably infinite with $I(X,Y)< \infty$, there exists a RV $U$ supported on $\mathcal{U}$ such that the leakage between $X$ and $U$ is equal to $\epsilon$, i.e., we have
 	\begin{align*}
 	I(U;X)= \epsilon,
 	\end{align*}
 	$Y$ is a deterministic function of $(U,X)$, i.e., we have 
 	\begin{align*}
 	H(Y|U,X)=0,
 	\end{align*}
 	$I(X;U|Y)$ can be  upper bounded as follows 
 	\begin{align*}
 	I(X;U|Y)\leq \alpha H(X|Y)+(1-\alpha)\left[ \log(I(X;Y)+1)+4\right],
 	\end{align*}
 	and 
 	$
 	|\mathcal{U}|\leq \left[|\mathcal{X}|(|\mathcal{Y}|-1)+2\right]\left[|\mathcal{X}|+1\right],
 	$
 	where $\alpha =\frac{\epsilon}{H(X)}$.
 \end{lemma}
 \begin{proof}
 	Let $\tilde{U}$ be the RV found by SFRL and $W$ be the same RV which is used to prove Lemma~\ref{lemma3}. It is sufficient to show that $I(X;U|Y)\leq \alpha H(X|Y)+(1-\alpha)\left[ \log(I(X;Y)+1)+4\right]$ since all other properties are already proved in Lemma~3. We have
 	\begin{align*}
 	I(X;\tilde{U},W|Y)&=I(X;\tilde{U}|Y)+I(X,W|\tilde{U},Y)\\&\stackrel{(a)}{=}I(X;\tilde{U}|Y)+\alpha H(X|\tilde{U},Y)\\&=I(X;\tilde{U}|Y)+\alpha(H(X|Y)-I(X;\tilde{U}|Y))\\&=\alpha H(X|Y)+(1-\alpha)I(X;\tilde{U}|Y)\\&\stackrel{(b)}{\leq} \!\alpha H(X|Y)\!+\!(1-\alpha)\!\left[ \log(I(X;Y)\!+\!1)\!+\!4\right],
 	\end{align*}
 	where in step (a) we used the fact that 
 	\begin{align*}
 	I(X,W|\tilde{U},Y) &= H(X|\tilde{U},Y)-H(X|W,\tilde{U},Y)\\&=
 	H(X|\tilde{U},Y)-(1-\alpha)H(X|\tilde{U},Y)\\&=\alpha H(X|\tilde{U},Y),
 	\end{align*}
 	and (b) follows since $\tilde{U}$ is produced by SFRL.
 \end{proof}
 In the next lemma, 
 we show that there exists a RV $U$ that satisfies \eqref{c1}, \eqref{c2} and has bounded entropy. The next lemma is a generalization of \cite[Lemma~2]{kostala} for dependent $X$ and $U$. 
\begin{lemma}
	For any pair of RVs $(X,Y)$ distributed according to $P_{XY}$ supported on alphabets $\mathcal{X}$ and $\mathcal{Y}$, where $|\mathcal{X}|$ is finite and $|\mathcal{Y}|$ is finite or countably infinite, there exists RV $U$ such that it satisfies \eqref{c1}, \eqref{c2}, and
	\begin{align*}
	H(U)\leq \sum_{x\in\mathcal{X}}H(Y|X=x)+\epsilon+h(\alpha)
	\end{align*}
	with $\alpha=\frac{\epsilon}{H(X)}$ and $h(\cdot)$ denotes the binary entropy function.
\end{lemma}
\begin{proof}
	Let $U=(\tilde{U},W)$ where $W$ is the same RV used in Lemma~\ref{lemma3} and $\tilde{U}$ is produced by FRL which has the same construction as used in proof of \cite[Lemma~1]{kostala}. Thus, by using \cite[Lemma~2]{kostala} we have
	\begin{align*}
	H(\tilde{U})\leq \sum_{x\in\mathcal{X}} H(Y|X=x),
	\end{align*} 
	therefore,
	\begin{align*}
	H(U)&=H(\tilde{U},W)\leq H(\tilde{U})+H(W),\\&\leq\sum_{x\in\mathcal{X}} H(Y|X=x)+H(W),
\end{align*}
where,
\begin{align*}
H(W)\! &= -(1-\alpha)\log(1-\alpha)\!-\!\!\sum_{x\in \mathcal{X}} \alpha P_X(x)\log(\alpha P_X(x)),\\&=h(\alpha)+\alpha H(X),
\end{align*}
which completes the proof.
\end{proof}
Before stating the next theorem we derive an expression for $I(Y;U)$. We have
\begin{align}
I(Y;U)&=I(X,Y;U)-I(X;U|Y),\nonumber\\&=I(X;U)+I(Y;U|X)-I(X;U|Y),\nonumber\\&=I(X;U)\!+\!H(Y|X)\!-\!H(Y|U,X)\!-\!I(X;U|Y).\label{key}
\end{align}
As argued in \cite{kostala}, \eqref{key} is an important observation to find lower and upper bounds for $h_{\epsilon}(P_{XY})$ and $g_{\epsilon}(P_{XY})$.
\begin{theorem}
	For any $0\leq \epsilon< I(X;Y)$ and pair of RVs $(X,Y)$ distributed according to $P_{XY}$ supported on alphabets $\mathcal{X}$ and $\mathcal{Y}$, if $h_{\epsilon}(P_{XY})>\epsilon$ then we have
	\begin{align*}
	H(Y|X)>0.
	\end{align*}
	Furthermore, if $H(Y|X)-H(X|Y)=H(Y)-H(X)>0$, then
	\begin{align*}
	h_{\epsilon}(P_{XY})>\epsilon.
	\end{align*}
\end{theorem}
\begin{proof}
	For proving the first part let $h_{\epsilon}(P_{XY})>\epsilon$. Using \eqref{key} we have
	\begin{align*}
	\epsilon&< h_{\epsilon}(P_{XY})\leq H(Y|X)+\sup_{U:I(X;U)\leq\epsilon} I(X;U)\\&=H(Y|X)+\epsilon
	\Rightarrow 0<H(Y|X).
	\end{align*}
	For the second part assume that $H(Y|X)-H(X|Y)>0$. Let $U$ be produced by EFRL. Thus, using the construction of $U$ as in Lemma~\ref{lemma3} we have $I(X,U)=\epsilon$ and $H(Y|X,U)=0$. Then by using \eqref{key} we obtain
	\begin{align*}
	h_{\epsilon}(P_{XY})&\geq \epsilon\!+\!H(Y|X)\!-\!H(X|Y)+H(X|Y,U)\\&\geq \epsilon\!+\!H(Y|X)\!-\!H(X|Y)>\epsilon.
	\end{align*}
\end{proof}
In the next theorem we provide a lower bound on $h_{\epsilon}(P_{XY})$.
\begin{theorem}\label{th.1}
	For any $0\leq \epsilon< I(X;Y)$ and pair of RVs $(X,Y)$ distributed according to $P_{XY}$ supported on alphabets $\mathcal{X}$ and $\mathcal{Y}$ we have
	\begin{align}\label{th2}
	h_{\epsilon}(P_{XY})\geq \max\{L_1^{\epsilon},L_2^{\epsilon},L_3^{\epsilon}\},
	\end{align}
	where
	\begin{align*}
	L_1^{\epsilon} &= H(Y|X)-H(X|Y)+\epsilon=H(Y)-H(X)+\epsilon,\\
	L_2^{\epsilon} &= H(Y|X)-\alpha H(X|Y)+\epsilon\\&\ -(1-\alpha)\left( \log(I(X;Y)+1)+4 \right),\\
	L_3^{\epsilon} &= \epsilon\frac{H(Y)}{I(X;Y)}+g_0(P_{XY})\left(1-\frac{\epsilon}{I(X;Y)}\right),
\end{align*}
and $\alpha=\frac{\epsilon}{H(X)}$.
The lower bound in \eqref{th2} is tight if $H(X|Y)=0$, i.e., $X$ is a deterministic function of $Y$. Furthermore, if the lower bound $L_1$ is tight then we have $H(X|Y)=0$. 
\end{theorem}
\begin{proof}
	$L_3^{\epsilon}$ can be derived by using \cite[Remark~2]{shahab}, since we have $h_{\epsilon}(P_{XY})\geq g_{\epsilon}(P_{XY})\geq L_3^{\epsilon}$. For deriving $L_1$, let $U$ be produced by EFRL. Thus, using the construction of $U$ as in Lemma~\ref{lemma3} we have $I(X,U)=\epsilon$ and $H(Y|X,U)=0$. Then, using \eqref{key} we obtain
	\begin{align*}
	h_{\epsilon}(P_{XY})&\geq I(U;Y)\\&=I(X;U)\!+\!H(Y|X)\!-\!H(Y|U,X)\!-\!I(X;U|Y)\\&=\epsilon+H(Y|X)-H(X|Y)+H(X|Y,U)\\ &\geq\epsilon+H(Y|X)-H(X|Y)=L_1.
	\end{align*} For deriving $L_2^{\epsilon}$, let $U$ be produced by ESFRL. Thus, by using the construction of $U$ as in Lemma~\ref{lemma4} we have $I(X,U)=\epsilon$, $H(Y|X,U)=0$ and $I(X;U|Y)\leq \alpha H(X|Y)+(1-\alpha)\left(\log(I(X;Y)+1)+4\right)$. Then, by using \eqref{key} we obtain
	\begin{align*}
	h_{\epsilon}(P_{XY})&\geq I(U;Y)\\&=I(X;U)\!+\!H(Y|X)\!-\!H(Y|U,X)\!-\!I(X;U|Y)\\&=\epsilon+H(Y|X)-I(X;U|Y)\\&\geq\epsilon+H(Y|X)-\alpha H(X|Y)\\&\ +(1-\alpha)\left(\log(I(X;Y)+1)+4\right)=L_2^{\epsilon}.
\end{align*}
Let $X$ be a deterministic function of $Y$. In this case, set $\epsilon=0$ in $L_1^{\epsilon}$ so that we obtain $h_0(P_{XY})\geq H(Y|X)$. Furthermore, by using \eqref{key} we have $h_0(P_{XY})\leq H(Y|X)$. Moreover, since $X$ is a deterministic function of $Y$, the Markov chain $X-Y-U$ holds and we have $h_0(P_{XY})=g_0(P_{XY})=H(Y|X)$. Therefore, $L_3^{\epsilon}$ can be rewritten as
\begin{align*}
L_3^{\epsilon}&=\epsilon\frac{H(Y)}{H(X)}+H(Y|X)\left(\frac{H(X)-\epsilon}{H(X)}\right),\\&=\epsilon\frac{H(Y)}{H(X)}+(H(Y)-H(X))\left(\frac{H(X)-\epsilon}{H(X)}\right),\\&=H(Y)-H(X)+\epsilon.
\end{align*}
$L_2^{\epsilon}$ can be rewritten as follows
\begin{align*}
L_2^{\epsilon}=H(Y|X)+\epsilon-(1-\frac{\epsilon}{H(X)})(\log(H(X)+1)+4).
\end{align*}
Thus, if $H(X|Y)=0$, then $L_1^{\epsilon}=L_3^{\epsilon}\geq L_2^{\epsilon}$. Now we show that $L_1^{\epsilon}=L_3^{\epsilon}$ is tight. By using \eqref{key} we have
\begin{align*}
I(U;Y) &\stackrel{(a)}{=} I(X;U)+H(Y|X)-H(Y|U,X),\\&\leq \epsilon+H(Y|X)=L_1^{\epsilon}=L_3^{\epsilon}.
\end{align*} 
where (a) follows since $X$ is deterministic function of $Y$ which leads to $I(X;U|Y)=0$. Thus, if $H(X|Y)=0$, the lower bound in \eqref{th2} is tight.
Now suppose that the lower bound $L_1^{\epsilon}$ is tight and $X$ is not a deterministic function of $Y$. Let $\tilde{U}$ be produced by FRL using the construction of \cite[Lemma~1]{kostala}. As argued in the proof of \cite[Th.~6]{kostala}, there exists $x\in\cal X$ and $y_1,y_2\in\cal Y$ such that $P_{X|\tilde{U},Y}(x|\tilde{u},y_1)>0$ and $P_{X|\tilde{U},Y}(x|\tilde{u},y_2)>0$ which results in $H(X|Y,\tilde{U})>0$. Let $U=(\tilde{U},W)$ where $W$ is defined in Lemma~\ref{lemma3}. For such $U$ we have
\begin{align*}
H(X|Y,U)&=(1-\alpha)H(X|Y,\tilde{U})>0,\\
\Rightarrow I(U;Y)&\stackrel{(a)}{=}\epsilon+H(Y|X)-H(X|Y)+H(X|Y,U)\\
&>\epsilon+H(Y|X)-H(X|Y).
\end{align*}
where in (a) we used the fact that such $U$ satisfies $I(X;U)=\epsilon$ and $H(Y|X,U)=0$. The last line is a contradiction with tightness of $L_1^{\epsilon}$, since we can achieve larger values, thus, $X$ needs to be a deterministic function of $Y$.
	\end{proof}
In next corollary we let $\epsilon=0$ and derive lower bound on $h_0(P_{XY})$.
\begin{corollary}\label{kooni}
	Let $\epsilon=0$. Then, for any pair of RVs $(X,Y)$ distributed according to $P_{XY}$ supported on alphabets $\mathcal{X}$ and $\mathcal{Y}$ we have
	\begin{align*}
	h_{0}(P_{XY})\geq \max\{L^0_1,L^0_2\},
	\end{align*}
	where 
	\begin{align*}
	L^0_1 &= H(Y|X)-H(X|Y)=H(Y)-H(X),\\
	L^0_2 &= H(Y|X) -\left( \log(I(X;Y)+1)+4 \right).\\
	\end{align*}
\end{corollary}
Note that the lower bound $L^0_1$ has been derived in \cite[Th.~6]{kostala}, while the lower bound $L^0_2$ is a new lower bound.
In the next two examples we compare the bounds $L_1^{\epsilon}$, $L_2^{\epsilon}$ and $L_3^{\epsilon}$ in special cases where $I(X;Y)=0$ and $H(X|Y)=0$.  
\begin{example}
	Let $X$ and $Y$ be independent. Then, we have
	\begin{align*}
	L_1^{\epsilon}&=H(Y)-H(X)+\epsilon,\\
	L_2^{\epsilon}&=H(Y)-\frac{\epsilon}{H(X)}H(X)+\epsilon-4(1-\frac{\epsilon}{H(X)}),\\
	&=H(Y)-4(1-\frac{\epsilon}{H(X)}).
	\end{align*}
	Thus,
	\begin{align*}
	L_2^{\epsilon}-L_1^{\epsilon}&=H(X)-4+\epsilon(\frac{4}{H(X)}-1),\\
	&=(H(X)-4)(1-\frac{\epsilon}{H(X)}).
	\end{align*}
  Consequently, for independent $X$ and $Y$ if $H(X)>4$, then $L_2^{\epsilon}>L_1^{\epsilon}$, i.e., the second lower bound is dominant and $h_{\epsilon}(P_{X}P_{Y})\geq L_2^{\epsilon}$. 
\end{example}
\begin{example}
	Let $X$ be a deterministic function of $Y$. As we have shown in Theorem~\ref{th.1}, if $H(X|Y)=0$, then
	\begin{align*}
	L_1^{\epsilon}&=L_3^{\epsilon}=H(Y|X)+\epsilon\\&\geq H(Y|X)+\epsilon-(1-\frac{\epsilon}{H(X)})(\log(H(X)+1)+4)\\&=L_2^{\epsilon}.
	\end{align*}
	Therefore, $L_1^{\epsilon}$ and $L_3^{\epsilon}$ become dominants.
\end{example}
In the next lemma we find a lower bound for $\sup_{U} H(U)$ where $U$ satisfies the leakage constraint $I(X;U)\leq\epsilon$, the bounded cardinality  and $ H(Y|U,X)=0$.
\begin{lemma}\label{koonimooni}
	For any pair of RVs $(X,Y)$ distributed according to $P_{XY}$ supported on alphabets $\mathcal{X}$ and $\mathcal{Y}$, then if $U$ satisfies $I(X;U)\leq \epsilon$, $H(Y|X,U)=0$ and $|\mathcal{U}|\leq \left[|\mathcal{X}|(|\mathcal{Y}|-1)+1\right]\left[|\mathcal{X}|+1\right]$, we have
	\begin{align*}
	\sup_U H(U)\!&\geq\! \alpha H(Y|X)\!+\!(1-\alpha)(\max_{x\in\mathcal{X}}H(Y|X=x))\!\\&+\!h(\alpha)\!+\!\epsilon \geq H(Y|X)\!+\!h(\alpha)\!+\!\epsilon,
	\end{align*}
	where $\alpha=\frac{\epsilon}{H(X)}$ and $h(\cdot)$ corresponds to the binary entropy.
\end{lemma} 
\begin{proof}
Let $U=(\tilde{U},W)$ where $W=\begin{cases}
X,\ \text{w.p}.\ \alpha\\
c,\ \ \text{w.p.}\ 1-\alpha
\end{cases}$, and $c$ is a constant which does not belong to the support of $X$, $Y$ and $\tilde{U}$, furthermore, $\tilde{U}$ is produced by FRL. Using \eqref{key} and \cite[Lemma~3]{kostala} we have
\begin{align}
H(\tilde{U}|Y)&=H(\tilde{U})-H(Y|X)+I(X;\tilde{U}|Y)\nonumber\\ &\stackrel{(a)}{\geq} \max_{x\in\mathcal{X}} H(Y|X=x)-H(Y|X)\nonumber\\&\ +H(X|Y)-H(X|Y,\tilde{U}),\label{toole}
\end{align} 	
where (a) follows from 
\cite[Lemma~3]{kostala}. Furthermore, in the first line we used $I(X;\tilde{U})=0$ and $H(Y|\tilde{U},X)=0$. Using \eqref{key} we obtain
\begin{align*}
H(U)&\stackrel{(a)}{=}\!H(U|Y)\!+\!H(Y|X)\!-\!H(X|Y)\!+\!\epsilon\!+\!H(X|Y,U),\\
&\stackrel{(b)}{=}H(W|Y)+\alpha H(\tilde{U}|Y,X)+(1-\alpha)H(\tilde{U}|Y)\\  & \ \ \ +H(Y|X)\!-\!H(X|Y)\!+\!\epsilon+(1\!-\!\alpha)H(X|Y,\tilde{U}),\\
&\stackrel{(c)}{=}\!(\alpha-1)H(X|Y)\!+\!h(\alpha)\!+\alpha H(\tilde{U}|Y,X)+\!\epsilon\\& \ \ +\!(1-\alpha)H(\tilde{U}|Y)\!+\!H(Y|X)\!+\!(1\!-\!\alpha)H(X|Y,\tilde{U}), \\
& \stackrel{(d)}{\geq}  (\alpha-1)H(X|Y)+h(\alpha)\!+\alpha H(\tilde{U}|Y,X)\\ \ \ \ &+\!\!(1-\alpha) (\max_{x\in\mathcal{X}}H(Y|X=x)-H(Y|X)+\!H(X|Y) \\& -\!H(X|Y,\tilde{U}))+H(Y|X)\!+\!\epsilon\!+\!(1-\alpha)H(X|Y,\tilde{U})\\
&=\alpha H(Y|X)+(1-\alpha)(\max_{x\in\mathcal{X}}H(Y|X=x))\\ &\ \ \ +h(\alpha)+\epsilon.
\end{align*}
In step (a) we used $I(U;X)=\epsilon$ and $H(Y|X,U)=0$ and in step (b) we used $H(U|Y)=H(W|Y)+H(\tilde{U}|Y,W)=H(W|Y)+\alpha H(\tilde{U}|Y,X)+(1-\alpha)H(\tilde{U}|Y)$ and $H(X|Y,U)=H(X|Y,\tilde{U},W)=(1-\alpha)H(X|Y,\tilde{U})$. In step (c) we used the fact that $P_{W|Y}=\begin{cases}
\alpha P_{X|Y}(x|\cdot)\ &\text{if}\ w=x,\\
1-\alpha \ &\text{if} \ w=c,
\end{cases}$
since $P_{W|Y}(w=x|\cdot)=\frac{P_{W,Y}(w=x,\cdot)}{P_Y(\cdot)}=\frac{P_{Y|W}(\cdot|w=x)P_W(w=x)}{P_Y(\cdot)}=\frac{P_{Y|X}(\cdot|x)\alpha P_X(x)}{P_Y(\cdot)}=\alpha P_{X|Y}(x|\cdot)$, furthermore, $P_{W|Y}(w=c|\cdot)=1-\alpha$. Hence, after some calculation we obtain $H(W|Y)=h(\alpha)+\alpha H(X|Y)$. Finally, step (d) follows from \eqref{toole}.
\end{proof}
\begin{remark}
	The constraint $|\mathcal{U}|\leq \left[|\mathcal{X}|(|\mathcal{Y}|-1)+1\right]\left[|\mathcal{X}|+1\right]$ in Lemma~\ref{koonimooni} guarantees that $\sup_U H(U)<\infty$.
\end{remark}
In the next lemma we find an upper bound for $h_{\epsilon}(P_{XY})$. 
\begin{lemma}\label{goh}
	For any $0\leq\epsilon< I(X;Y)$ and pair of RVs $(X,Y)$ distributed according to $P_{XY}$ supported on alphabets $\mathcal{X}$ and $\mathcal{Y}$ we have
	\begin{align*}
	g_{\epsilon}(P_{XY})\leq h_{\epsilon}(P_{XY})\leq H(Y|X)+\epsilon.
	\end{align*}
\end{lemma}
\begin{proof}
	By using \eqref{key} we have
	\begin{align*}
	h_{\epsilon}(P_{XY})\leq H(Y|X)+\sup I(U;X)\leq H(Y|X)+\epsilon. 
	\end{align*}
\end{proof}
\begin{corollary}
	If $X$ is a deterministic function of $Y$, then by using Theorem~\ref{th.1} and Lemma~\ref{goh} we have
	\begin{align*}
	g_{\epsilon}(P_{XY})=h_{\epsilon}(P_{XY})=H(Y|X)+\epsilon,
	\end{align*}
	since in this case the Markov chain $X-Y-U$ holds.
\end{corollary}
\begin{lemma}\label{kir}
	Let $\bar{U}$ be an optimizer of $h_{\epsilon}(P_{XY})$. We have
	\begin{align*}
	H(Y|X,\bar{U})=0.
	\end{align*}
\end{lemma}
\begin{proof}
	The proof is similar to \cite[Lemma~5]{kostala}. Let $\bar{U}$ be an optimizer of $h_{\epsilon}(P_{XY})$ and assume that $H(Y|X,\bar{U})>0$. Consequently, we have $I(X;\bar{U})\leq \epsilon.$ Let $U'$ be founded by FRL  with $(X,\bar{U})$ instead of $X$ in Lemma~\ref{lemma1} and same $Y$, that is $I(U';X,\bar{U})=0$ and $H(Y|X,\bar{U},U')=0$. Using \cite[Th.~5]{kostala} we have
	\begin{align*}
	I(Y;U')>0,
	\end{align*}
	since we assumed $H(Y|X,\bar{U})>0$. Let $U=(\bar{U},U')$ and we first show that $U$ satisfies $I(X;U)\leq \epsilon$. We have
	\begin{align*}
	I(X;U)&=I(X;\bar{U},U')=I(X;\bar{U})+I(X;U'|\bar{U}),\\
	&=I(X;\bar{U})+H(U'|\bar{U})-H(U'|\bar{U},X),\\
	&=I(X;\bar{U})+H(U')-H(U')\leq \epsilon,
	\end{align*}
	where in last line we used the fact that $U'$ is independent of the pair $(X,\bar{U})$. Finally, we show that $I(Y;U)>I(Y,\bar{U})$ which is a contradiction with optimality of $\bar{U}$. We have
	\begin{align*}
	I(Y;U)&=I(Y;\bar{U},U')=I(Y;U')+I(Y;\bar{U}|U'),\\
	&=I(Y;U')+I(Y,U';\bar{U})-I(U';\bar{U})\\
	&= I(Y;U')+I(Y,\bar{U})+I(U';\bar{U}|Y)-I(U';\bar{U})\\
	&\stackrel{(a)}{\geq} I(Y;U')+I(Y,\bar{U})\\
	&\stackrel{(b)}{>} I(Y,\bar{U}),
	\end{align*}
	where in (a) follows since $I(U';\bar{U}|Y)\geq 0$ and $I(U';\bar{U}=0$. Step (b) follows since $I(Y;U')>0$. Thus, the obtained contradiction completes the proof.
\end{proof}
In the next theorem we generalize the equivalent statements in \cite[Th.~7]{kostala} for bounded leakage between $X$ and $U$.
\begin{theorem}
	For any $\epsilon<I(X;Y)$, we have the following equivalencies
	\begin{itemize}
		\item [i.] $g_{\epsilon}(P_{XY})=H(Y|X)+\epsilon$,
		\item [ii.] $g_{\epsilon}(P_{XY})=h_{\epsilon}(P_{XY})$,
		\item [iii.] $h_{\epsilon}(P_{XY})=H(Y|X)+\epsilon$.
	\end{itemize}
\end{theorem}
\begin{proof}
	\begin{itemize}
		\item i $\Rightarrow$ ii: Using Lemma~\ref{goh} we have $H(Y|X)+\epsilon= g_{\epsilon}(P_{XY})\leq h_{\epsilon}(P_{XY}) \leq H(Y|X)+\epsilon$. Thus, $g_{\epsilon}(P_{XY})=h_{\epsilon}(P_{XY})$.
		\item ii $\Rightarrow$ iii: Let $\bar{U}$ be the optimizer of $g_{\epsilon}(P_{XY})$. Thus, the Markov chain $X-Y-\bar{U}$ holds and we have $I(X;U|Y)=0$. Furthermore, since $g_{\epsilon}(P_{XY})=h_{\epsilon}(P_{XY})$ this $\bar{U}$ achieves $h_{\epsilon}(P_{XY})$. Thus, by using Lemma~\ref{kir} we have $H(Y|\bar{U},X)=0$ and according to \eqref{key} 
		\begin{align}
		I(\bar{U};Y)&=I(X;\bar{U})\!+\!H(Y|X)\!-\!H(Y|\bar{U},X)\!\\& \ \ \ -\!I(X;\bar{U}|Y)\nonumber\\
		&=I(X;\bar{U})\!+\!H(Y|X).\label{kir2}
		\end{align}
		We claim that $\bar{U}$ must satisfy $I(X;Y|\bar{U})>0$ and $I(X;\bar{U})=\epsilon$. For the first claim assume that $I(X;Y|\bar{U})=0$, hence the Markov chain $X-\bar{U}-Y$ holds. Using $X-\bar{U}-Y$ and $H(Y|\bar{U},X)=0$ we have $H(Y|\bar{U})=0$, hence $Y$ and $\bar{U}$ become independent. Using \eqref{kir2}
		\begin{align*}
		H(Y)&=I(Y;\bar{U})=I(X;\bar{U})\!+\!H(Y|X),\\
		&\Rightarrow I(X;\bar{U})=I(X;Y).
		\end{align*}
		The last line is a contradiction since by assumption we have $I(X;\bar{U})\leq \epsilon < I(X;Y)$. Thus, $I(X;Y|\bar{U})>0$.
		 For proving the second claim assume that $I(X;\bar{U})=\epsilon_1<\epsilon$. Let $U=(\bar{U},W)$ where $W=\begin{cases}
		 Y,\ \text{w.p}.\ \alpha\\
		 c,\ \ \text{w.p.}\ 1-\alpha
		 \end{cases}$,
		 and $c$ is a constant that $c\notin \mathcal{X}\cup \mathcal{Y}\cup \mathcal{\bar{U}}$ and $\alpha=\frac{\epsilon-\epsilon_1}{I(X;Y|\bar{U})}$.
		 We show that $\frac{\epsilon-\epsilon_1}{I(X;Y|\bar{U})}<1$. By the assumption we have
		 \begin{align*}
		 \frac{\epsilon-\epsilon_1}{I(X;Y|\bar{U})}<\frac{I(X;Y)-I(X;\bar{U})}{I(X;Y|\bar{U})}\stackrel{(a)}{\leq}1,
		 \end{align*}
		 step (a) follows since $I(X;Y)-I(X;\bar{U})-I(X;Y|\bar{U})=I(X;Y)-I(X;Y,\bar{U})\leq 0$. It can be seen that such $U$ satisfies $H(Y|X,U)=0$ and $I(X;U|Y)=0$ since
		 \begin{align*}
		 H(Y|X,U)&=\alpha H(Y|X,\bar{U},Y)\\ & \ \ \ + (1-\alpha)H(Y|X,\bar{U})=0,\\
		 I(X;U|Y)&= H(X|Y)-H(X|Y,\bar{U},W)\\&=H(X|Y)\!-\!\alpha H(X|Y,\bar{U})\!\\& \ \ \ -\!(1\!-\!\alpha) H(X|Y,\bar{U})\\&=H(X|Y)-H(X|Y)=0, 
		 \end{align*}
		 where in deriving the last line we used the Markov chain $X-Y-\bar{U}$. Furthermore, 
		 \begin{align*}
		 I(X;U)&=I(X;\bar{U},W)=I(X;\bar{U})+I(X;W|\bar{U})\\
		 &=I(X;\bar{U})+\alpha H(X|\bar{U})-\alpha H(X|\bar{U},Y)\\
		 &=I(X;\bar{U})+\alpha I(X;Y|\bar{U})\\
		 &=\epsilon_1+\epsilon-\epsilon_1=\epsilon,
		 \end{align*}
		 and
		 \begin{align*}
		 I(Y;U)&=I(X;U)\!+\!H(Y|X)\!-\!H(Y|U,X)\\& \ \ \ -\!I(X;U|Y)\\
		 &=\epsilon+H(Y|X).
		 \end{align*}
		 Thus, if $I(X;\bar{U})=\epsilon_1<\epsilon$ we can substitute $\bar{U}$ by $U$ for which $I(U;Y)>I(\bar{U};Y)$. This is a contraction and we conclude that $I(X;\bar{U})=\epsilon$ which proves the second claim. Hence, \eqref{kir2} can be rewritten as
		 \begin{align*}
		 I(\bar{U};Y)=\epsilon+H(Y|X).
		 \end{align*}
		 As a result $h_\epsilon(P_{XY})=\epsilon+H(Y|X)$ and the proof is completed.
		\item iii $\Rightarrow$ i: Let $\bar{U}$ be the optimizer of $h_{\epsilon}(P_{XY})$ and $h_{\epsilon}(P_{XY})=H(Y|X)+\epsilon$. Using Lemma~\ref{kir} we have $H(Y|\bar{U},X)=0$. By using \eqref{key} we must have 
		$I(X;\bar{U}|Y)=0$ and $I(X;\bar{U})=\epsilon$. We conclude that for this $\bar{U}$, the Markov chain $X-Y-\bar{U}$ holds and as a result $\bar{U}$ achieves $g_{\epsilon}(P_{XY})$ and we have $g_{\epsilon}(P_{XY})=H(Y|X)+\epsilon$. 
	\end{itemize}
\end{proof}
\subsection* {Special case: $\epsilon=0$ (Independent $X$ and $Y$)}
In this section we derive new lower and upper bounds for $h_0(P_{XY})$ and compare them with the previous bounds found in \cite{kostala}. We first state the definition of \emph{excess functional information} defined in
\cite{kosnane} as
\begin{align*}
\psi(X\rightarrow Y)=\inf_{\begin{array}{c} 
	\substack{P_{U|Y,X}: I(U;X)=0,\ H(Y|X,U)=0}
	\end{array}}I(X;U|Y),
\end{align*}
and the lower bound on $\psi(X\rightarrow Y)$ derived in \cite[Prop.~1]{kosnane} is given in the next lemma. Since this lemma is useful for deriving the lower bound we state it here.
\begin{lemma}\cite[Prop.~1]{kosnane}\label{haroomi}
	For discrete $Y$ we have 
	\begin{align}
	&\psi(X\rightarrow Y)\geq\nonumber\\& -\sum_{y\in\mathcal{Y}}\!\int_{0}^{1}\!\!\! \mathbb{P}_X\{P_{Y|X}(y|X)\geq t\}\log (\mathbb{P}_X\{P_{Y|X}(y|X)\geq t\})dt\nonumber\\&-I(X;Y),\label{koonsher}
	\end{align}
	where for $|\mathcal{Y}|=2$ the equality holds and it is attained by the Poisson functional representation in \cite{kosnane}.
\end{lemma}
\begin{remark}
	The lower bound in \eqref{koonsher} can be negative. For instance, let $Y$ be a deterministic function of $X$, i.e., $H(Y|X)=0$. In this case we have $-\sum_{y\in\mathcal{Y}}\!\int_{0}^{1} \mathbb{P}_X\{P_{Y|X}(y|X)\geq t\}\log (\mathbb{P}_X\{P_{Y|X}(y|X)\geq t\})dt-I(X;Y)=-I(X;Y)=-H(Y).$
\end{remark}
In the next theorem lower and upper bounds on $h_0(P_{XY})$ are provided.
\begin{theorem}
	For any pair of RVs $(X,Y)$ distributed according to $P_{XY}$ supported on alphabets $\mathcal{X}$ and $\mathcal{Y}$ we have
	\begin{align*}
	\max\{L^0_1,L^0_2\} \leq h_{0}(P_{XY})\leq \min\{U^0_1,U^0_2\},
	\end{align*}
	where $L^0_1$ and $L^0_2$ are defined in Corollary~\ref{kooni} and 
	\begin{align*}
	&U^0_1 = H(Y|X),\\
	&U^0_2 =  H(Y|X) +\sum_{y\in\mathcal{Y}}\int_{0}^{1} \mathbb{P}_X\{P_{Y|X}(y|X)\geq t\}\times\\&\log (\mathbb{P}_X\{P_{Y|X}(y|X)\geq t\})dt+I(X;Y).
	\end{align*}	
	Furthermore, if $|\mathcal{Y}|=2$, then we have 
	\begin{align*}
	h_{0}(P_{XY}) = U^0_2.
	\end{align*}
\end{theorem}
\begin{proof}
	 $L^0_1$ and $L^0_2$ can be obtained by letting $\epsilon=0$ in Theorem~\ref{th.1}. $U^0_1$ which has been derived in \cite[Th.~7]{kostala} can be obtained by \eqref{key}. $U^0_1$ can be derived as follows. Since $X$ and $U$ are independent, \eqref{key} can be rewritten as 
	 \begin{align*}
	 I(Y;U)=H(Y|X)-H(Y|U,X)-I(X;U|Y),
	 \end{align*} 
	 thus, using Lemma~\ref{haroomi}
	 \begin{align*}
	 h_0(P_{XY})&\leq H(Y|X)-\inf_{H(Y|U,X)=0,\ I(X;U)=0}I(X;U|Y)\\ &= H(Y|X)-\psi(X\rightarrow Y)\\&\leq H(Y|X)\\&+\sum_{y\in\mathcal{Y}}\int_{0}^{1} \mathbb{P}_X\{P_{Y|X}(y|X)\geq t\}\times\\&\log (\mathbb{P}_X\{P_{Y|X}(y|X)\geq t\})dt+I(X;Y). 	 \end{align*}
	 For $|\mathcal{Y}|=2$ using Lemma~\ref{haroomi} we have $\psi(X\rightarrow Y)=-\sum_{y\in\mathcal{Y}}\int_{0}^{1} \mathbb{P}_X\{P_{Y|X}(y|X)\geq t\}\log (\mathbb{P}_X\{P_{Y|X}(y|X)\geq t\})dt-I(X;Y)$ and let $\bar{U}$ be the RV that attains this bound. Thus,
	 \begin{align*}
	 I(\bar{U};Y)&=H(Y|X)\\&+\sum_{y\in\mathcal{Y}}\int_{0}^{1} \mathbb{P}_X\{P_{Y|X}(y|X)\geq t\}\times\\&\log (\mathbb{P}_X\{P_{Y|X}(y|X)\geq t\})dt+I(X;Y).
	 \end{align*}
	 Therefore, $\bar{U}$ attains $U_2^0$ and $h_0(P_{XY})=U_0^2$.
\end{proof}
	As mentioned before the upper bound $U^0_1$ has been derived in \cite[Th.~7]{kostala}. The upper bound $U^0_2$ is a new upper bound.
\begin{lemma}\label{ankhar}
	If $X$ is a deterministic function of $Y$, i.e., $H(X|Y)=0$, we have
	\begin{align*}
	&\sum_{y\in\mathcal{Y}}\int_{0}^{1} \mathbb{P}_X\{P_{Y|X}(y|X)\geq t\}\log (\mathbb{P}_X\{P_{Y|X}(y|X)\geq t\})dt\\&+I(X;Y)=0.
	\end{align*}
\end{lemma}
\begin{proof}
	Since $X$ is a deterministic function of $Y$, for any $y\in \cal Y$ we have 
	\begin{align*}
	P_{Y|X}(y|x)=\begin{cases}
	\frac{P_Y(y)}{P_X(x)},\ &x=f(y)\\
	0, \ &\text{else} 
	\end{cases},
	\end{align*}
	thus,
	\begin{align*}
	&\sum_{y\in\mathcal{Y}}\int_{0}^{1} \mathbb{P}_X\{P_{Y|X}(y|X)\geq t\}\log (\mathbb{P}_X\{P_{Y|X}(y|X)\geq t\})dt\\
	&=	\sum_{y\in\mathcal{Y}}\!\int_{0}^{\frac{P_Y(y)}{P_X(x=f(y))}} \!\!\!\!\!\!\!\!\!\!\!\!\!\!\!\!\!\!\!\!\!\!\!\!\!\mathbb{P}_X\{P_{Y|X}(y|X)\geq t\}\log (\mathbb{P}_X\{P_{Y|X}(y|X)\geq t\})dt\\
	&= \sum_{y\in\mathcal{Y}} \frac{P_Y(y)}{\mathbb{P}_X\{x=f(y)\}}\mathbb{P}_X\{x=f(y)\}\log(\mathbb{P}_X\{x=f(y)\})\\
	&= \sum_{y\in\mathcal{Y}} P_Y(y)\log(\mathbb{P}_X\{x=f(y)\})\\
	&= \sum_{y\in\mathcal{Y}} P_X(x)\log(P_X(x))=-H(X)=-I(X;Y),
	\end{align*}
	where in last line we used $\sum_{y\in\mathcal{Y}} P_Y(y)\log(\mathbb{P}_X\{x=f(y)\})=\sum_{x\in\mathcal{X}} \sum_{y:x=f(y)} P_Y(y)\log(\mathbb{P}_X\{x=f(y)\})=\sum_{x\in\mathcal{X}} P_X(x)\log(P_X(x))$.
\end{proof}
\begin{remark}
	According to Lemma~\ref{ankhar}, if $X$ is a deterministic function of $Y$, then we have $U_2^0=U_1^0$.  
\end{remark}
In the next example we compare the bounds $U^0_1$ and $U^0_2$ for a $BSC(\theta)$.
\begin{figure}[]
	\centering
	\includegraphics[scale = .2]{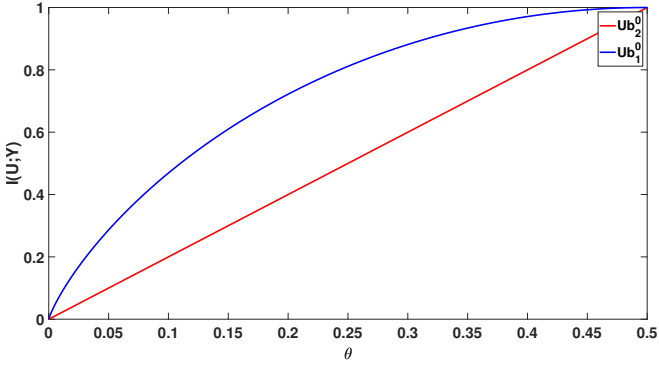}
	\caption{Comparing the upper bounds $U_1^0$ and $U_2^0$ for $BSC(\theta)$. The blue curve illustrates the upper bound found in \cite{kostala} and the red line shows the upper bound found in this work.}
	\label{fig:kir}
\end{figure}
\begin{example}(Binary Symmetric Channel)
   Let the binary RVs $X\in\{0,1\}$ and $Y\in\{0,1\}$ have the following joint distribution
   \begin{align*}
   P_{XY}(x,y)=\begin{cases}
   \frac{1-\theta}{2}, \ &x=y\\
   \frac{\theta}{2}, \ &x\neq y
   \end{cases},
   \end{align*}
   where $\theta<\frac{1}{2}$. We obtain
   \begin{align*}
   &\sum_{y\in\mathcal{Y}}\int_{0}^{1} \mathbb{P}_X\{P_{Y|X}(y|X)\geq t\}\log (\mathbb{P}_X\{P_{Y|X}(y|X)\geq t\})dt\\&=\int_{\theta}^{1-\theta} \mathbb{P}_X\{P_{Y|X}(0|X)\geq t\}\log (\mathbb{P}_X\{P_{Y|X}(0|X)\geq t\})dt\\&+\int_{\theta}^{1-\theta} \mathbb{P}_X\{(P_{Y|X}(1|X)\geq t\}\log (\mathbb{P}_X\{P_{Y|X}(1|X)\geq t\})dt\\&=(1-2\theta)\left(P_X(0)\log (P_X(0))+P_X(1)\log (P_X(1))\right)\\&=-(1-2\theta)H(X)=-(1-2\theta).
   \end{align*}
   Thus,
   \begin{align*}
   U_2^0&=H(Y|X) +\sum_{y\in\mathcal{Y}}\int_{0}^{1} \mathbb{P}_X\{P_{Y|X}(y|X)\geq t\}\times\\&\log (\mathbb{P}_X\{P_{Y|X}(y|X)\geq t\})dt+I(X;Y)\\&=h(\theta)-(1-2\theta)+(1-h(\theta))=2\theta,\\
   U_1^0&=H(Y|X)=h(\theta),
   \end{align*}
   where $h(\cdot)$ corresponds to the binary entropy function. As shown in Fig.~ \ref{fig:kir}, we have
   \begin{align*}
   h_{0}(P_{XY})\leq U^0_2\leq U^0_1.
   \end{align*}
\end{example}
\begin{example}(Erasure Channel)
	Let the RVs $X\in\{0,1\}$ and $Y\in\{0,e,1\}$ have the following joint distribution
	\begin{align*}
	P_{XY}(x,y)=\begin{cases}
	\frac{1-\theta}{2}, \ &x=y\\
	\frac{\theta}{2}, \ &y=e\\
	0, \ & \text{else}
	\end{cases},
	\end{align*}
	where $\theta<\frac{1}{2}$. We have
	\begin{align*}
	&\sum_{y\in\mathcal{Y}}\int_{0}^{1} \mathbb{P}_X\{P_{Y|X}(y|X)\geq t\}\log (\mathbb{P}_X\{P_{Y|X}(y|X)\geq t\})dt\\&=\int_{0}^{1-\theta} \mathbb{P}_X\{P_{Y|X}(0|X)\geq t\}\log (\mathbb{P}_X\{P_{Y|X}(0|X)\geq t\})dt\\&+\int_{0}^{1-\theta} \mathbb{P}_X\{P_{Y|X}(1|X)\geq t\}\log (\mathbb{P}_X\{P_{Y|X}(1|X)\geq t\})dt\\&=-(1-\theta)H(X)=-(1-\theta).
	\end{align*}
	Thus, 
	\begin{align*}
	U_2^0&=H(Y|X) +\sum_{y\in\mathcal{Y}}\int_{0}^{1} \mathbb{P}_X\{P_{Y|X}(y|X)\geq t\}\times\\&\log (\mathbb{P}_X\{P_{Y|X}(y|X)\geq t\})dt+I(X;Y)\\&=h(\theta)-(1-\theta)+h(\theta)+1-\theta-h(\theta) \\&=h(\theta),\\
	 U_1^0&=H(Y|X)=h(\theta). 
	\end{align*}
	Hence, in this case, $U_1^0=U_2^0=h(\theta)$. Furthermore, in \cite[Example~8]{kostala}, it has been shown that for this pair of $(X,Y)$ we have $g_0(P_{XY})=h_0(P_{XY})=h(\theta)$.
\end{example}
In \cite[Prop.~2]{kosnane} it has been shown that for every $\alpha\geq 0$, there exist a pair $(X,Y)$ such that $I(X;Y)\geq \alpha$ and 
\begin{align}
\psi(X\rightarrow Y)\geq \log(I(X;Y)+1)-1.\label{kir1}
\end{align}
\begin{lemma}\label{choon}
	Let $(X,Y)$ be as in \cite[Prop.~2]{kosnane}, i.e. $(X,Y)$ satisfies \eqref{kir1}. Then for such pair we have
	\begin{align*}
	&H(Y|X)-\log(I(X;Y)+1)-4\\&\leq h_0(P_{XY})\leq H(Y|X)-\log(I(X;Y)+1)+1. 
	\end{align*}
	\begin{proof}
		The lower bound follows from Corollary~1. For the upper bound, we use \eqref{key} and \eqref{kir1} so that
		\begin{align*}
		I(U;Y)&\leq H(Y|X)-\psi(X\rightarrow Y)\\ &\leq H(Y|X)-\log(I(X;Y)+1)+1.
		\end{align*}
	\end{proof}
\end{lemma}
\begin{remark}
	From Lemma~\ref{choon} and Corollary~1 we can conclude that the lower bound $L_2^0=H(Y|X)-(\log(I(X;Y)+1)+4)$ is tight within $5$ bits.
\end{remark}
\section{conclusion}\label{concull}
It has been shown that by extending the FRL and SFRL, upper bound for $h_{\epsilon}(P_{XY})$ and $g_{\epsilon}(P_{XY})$ and lower bound for $h_{\epsilon}(P_{XY})$ can be derived. 
If $X$ is a deterministic function of $Y$, then the bounds are tight. Moreover, a necessary condition for an optimizer of $h_{\epsilon}(P_{XY})$ has been obtained. In the case of perfect privacy, new lower and upper bounds are derived using ESFRL and excess functional information. In an example it has been shown that new bounds are dominant compared to the previous bounds.

\bibliographystyle{IEEEtran}
\bibliography{IEEEabrv,IZS}
\end{document}